\documentclass[11pt, a4paper]{article}

\pdfoutput=1

\usepackage[utf8]{inputenc}
\usepackage[hidelinks]{hyperref}
\usepackage{microtype}
\usepackage{xcolor}

\usepackage{mathtools, amsthm, amsfonts, amsmath, amssymb}
\usepackage{booktabs}
\usepackage{braket}
\usepackage{commath}
\usepackage{caption}
\usepackage{titlesec}
\titlelabel{\thetitle. }
\titleformat*{\section}{\large\bfseries}

\usepackage[capitalise]{cleveref}

\newtheorem{theorem}{Theorem}
\newtheorem{proposition}[theorem]{Proposition}

\theoremstyle{definition}

\newtheorem*{problem*}{Problem}

\newtheorem{remark}[theorem]{Remark}

\newcommand{\gengrp}{\Gamma}
\newcommand{\genHam}{H}
\newcommand{\binent}{h}
\DeclareMathOperator{\lovtheta}{\vartheta} 
\DeclareMathOperator{\symork}{\xi_{\textnormal{sym}}} 
\DeclareMathOperator{\ork}{\xi} 
\newcommand{\bset}[1]{\{0,1\}^{#1}}
\DeclareMathOperator{\supp}{supp}

\newcommand{\CC}{\mathbb{C}}

\newcommand{\NN}{\mathbb{N}}

\newcommand{\RR}{\mathbb{R}}

\newcommand{\Cay}{\mathrm{Cay}}

\begin{document}

\title{\bf\Large On the orthogonal rank of Cayley graphs and impossibility of quantum round elimination}
\author{Jop Bri\"{e}t\thanks{
QuSoft, CWI, Science Park 123, 1098 XG Amsterdam, Netherlands. Email: j.briet@cwi.nl.
Supported by a VENI grant from the Netherlands Organisation for Scientific Research~(NWO).
} 
\and 
Jeroen Zuiddam\thanks{
QuSoft, CWI and University of Amsterdam, Science Park 123, 1098 XG Amsterdam, Netherlands. Email: j.zuiddam@cwi.nl.
Supported by NWO through the research programme 617.023.116.
}
}

\maketitle

\begin{abstract}
After Bob sends Alice a bit, she responds with a lengthy reply.
At the cost of a factor of two in the total communication, Alice could just as well have given the two possible replies without listening and have Bob select which applies to him.
Motivated by a conjecture stating that this form of ``round elimination'' is impossible in exact quantum communication complexity, we study the orthogonal rank and a symmetric variant thereof
for a certain family of Cayley graphs. 
The orthogonal rank of a graph is the smallest number $d$ for which one can label each vertex with a nonzero $d$-dimensional complex vector such that adjacent vertices receive orthogonal vectors. 

We show an $\exp(n)$ lower bound on the orthogonal rank of the graph on~$\{0,1\}^n$ in which two strings are adjacent if they have Hamming distance at least~$n/2$.
In combination with previous work, this implies an affirmative answer to the above conjecture.
\end{abstract}

\section{Introduction}

\paragraph{The orthogonal rank of Cayley graphs.}

In the following all graphs are simple and undirected.
For a graph $G = (V,E)$, a map $\phi: V \to \CC^d$ is an \emph{orthogonal embedding in dimension $d$} if $\langle \phi(v), \phi(v)\rangle = 1$ for all $v\in V$ and $\langle\phi(v),\phi(w)\rangle = 0$ for all $(v,w) \in E$. 
The \emph{orthogonal rank} of $G$, denoted $\ork(G)$, is the smallest positive integer $d$ such that there is an orthogonal embedding of $G$ in dimension~$d$. 
Here we prove bounds on the orthogonal ranks of certain Cayley graphs.

For a finite group~$\Gamma$ and a subset $S\subseteq \Gamma$, the \emph{Cayley graph} $G = \Cay(\Gamma,S)$ is the graph with vertex set~$V = \Gamma$ and in which $g,h\in\Gamma$ are connected by an edge if and only if $gh^{-1} \in S$ or $hg^{-1} \in S$. 
We shall also be interested in the following variant of the orthogonal rank.
Say that an orthogonal embedding $\phi: \Gamma \to \CC^d$ is \emph{symmetric} if there exists a function $f: \Gamma \to \CC$ such that $\langle \phi(g), \phi(h)\rangle = f(gh^{-1})$ for all $g,h\in \Gamma$. 
Then, the \emph{symmetric orthogonal} rank of~$G$, denoted~$\symork(G)$, is defined as the smallest positive integer~$d$ such that there exists a symmetric orthogonal embedding in dimension~$d$. 
Note that, clearly, $\ork(G)\leq\symork(G)$.

Our main results concern bounds on the orthogonal rank and symmetric orthogonal rank of certain Cayley graphs based on powers of cyclic groups. For a positive integer~$m$, let $C_m = \{0,1,\ldots, m-1\}$ be the cyclic group of~$m$ elements. For a positive integer $k$, let $[k]$ denote the set $\{1, 2, \ldots, k\}$. 
For a positive integer~$n$ and parameter $d\in[(m-1)n]$, define $\genHam_m^n(d)$ to be the Cayley graph $\Cay(C_m^{\times n}, S$) with
 $S = \{x \in C_m^{\times n} : \sum_i x_i \geq d\}$, where in the summation we consider the elements $x_i$ as elements in $\NN$.
Then, the graph $H_2^n(d)$ is precisely the graph with vertex set~$\bset{n}$ where two strings form an edge if and only if they have Hamming distance at least~$d$.

\paragraph{Quantum communication complexity.}
The orthogonal rank in general is a poorly understood parameter.
Much impetus for its study has recently come from quantum information theory, in particular in the context of quantum entanglement~\cite{Cameron:2007}.
The problem of bounding the orthogonal rank of the above-mentioned Cayley graph $\genHam_2^n(d)$ arose from a question in  exact quantum communication complexity.
Here two parties, Alice and Bob, receive inputs $x,y$ from sets $\mathcal X,\mathcal Y$, respectively, 
and their goal is to compute a function~$f(x,y)$ depending on both of their inputs, using as little communication as possible.
In a \emph{promise problem}, the inputs are guaranteed to be drawn from a subset $\mathcal D\subseteq \mathcal X\times \mathcal Y$ known to the parties in advance.

In the most general deterministic classical protocol the parties take turns sending each other binary sequences until both of them know the answer.
The (exact) \emph{communication complexity} is defined as the minimum number of bits sent back and forth in such a protocol under worst-case inputs.
The \emph{one-round} communication complexity is the length of the shortest string Alice can send Bob so that he can learn the answer under worst-case inputs.

In the quantum setting, the parties may send each other \emph{qubits} instead of classical bits, which gives them at least as much power as in the classical setting and is well-known to sometimes lead to dramatic savings~\cite{Buhrman:1998} (see~\cite{nielsen2010quantum, wilde2013quantum} for an introduction to quantum information theory).
The \emph{quantum communication complexity} is defined analogously to its classical counterpart.
The \emph{one-round} quantum communication complexity turns out to be characterized precisely by the orthogonal rank of the graph $G = (V, E)$ with vertex set $V = \mathcal X$ and where $u,v\in V$ form an edge if there exists a $y\in \mathcal Y$ such that $(u,y)\in \mathcal D$ and $(v,y)\in\mathcal D$.
De Wolf~\cite[Theorem~8.5.2]{deWolf:PhD} showed that the one-round quantum communication complexity equals~$\lceil\log_2 \ork(G)\rceil$.

\paragraph{Impossibility of quantum round elimination.}
\emph{Round elimination} is a basic procedure whereby the number of rounds of communication is reduced at the cost of some additional communication.
For example, if Bob is supposed to send Alice a single bit after which she is supposed to reply with a $k$-bit string, she could just as well immediately send Bob the two $k$-bit strings corresponding to the bits he might have sent, after which he picks out the appropriate one. 
This removes one round of communication at the cost of roughly a factor of two increase in the total number of communicated bits.

It was conjectured in~\cite{briet} that a quantum analogue of  round elimination is impossible, in the sense that removing a round can  result in unexpectedly large increases in (quantum) communication.
The following promise problem was suggested as a possible candidate to show this:
Alice is given an $n$-bit string~$x$ and Bob is given a set $Y\subseteq\bset{n}$ containing~$x$ such that for some integer $d\geq n/2$ the pairwise Hamming distances between the strings in~$Y$ equals~$d$.
The parties' goal is for Bob to learn~$x$, that is $f(x,Y) = x$.
The authors gave a two-round protocol  for this problem in which Bob first sends Alice a single qubit, after which Alice replies with a $k = \lceil\log_2 n + 1\rceil$-qubit sequence. 
The naive analogue of the above   round-elimination example would say that there is a one-round $2k$-qubit protocol.
However, in~\cite{briet} it was conjectured that
the orthogonal rank of the graph associated to the problem, the graph $\genHam_2^n(n/2)$, is of the order~$n^{\omega(1)}$, implying that the
one-round quantum communication complexity is in fact~$\omega(k)$.
It was shown that for $n$ even, $2n \leq \ork(\genHam_2^n(n/2)) \leq 2^{\binent(1/4)n + 1} \approx 2^{0.81 n}$, where $\binent(p) = -p\log_2 p - (1-p)\log_2(1-p)$ is the binary entropy function.

Towards resolving this conjecture, it was suggested by Buhrman~\cite{buhrman:personal} to examine the potentially easier problem of determining the symmetric orthogonal rank of~$\genHam_2^n(n/2)$.
Here we determine this parameter exactly for the more general class of Cayley graphs described above.

\begin{theorem}\label{thm:symmetric} For all positive integers~$m,n$ and any $d\in[(m-1)n]$ that is divisible by $m-1$, we have
\[
\symork(\genHam_m^n(d)) = m^{n-\frac{d}{m-1}}.
\]
\end{theorem}


Observe that the above result improves on the upper bound of the orthogonal rank of~\cite{briet}, since  $\ork(\genHam_2^n(n/2)) \leq \symork(\genHam_2^n(n/2)) \leq 2^{n/2}$.

More importantly, the main conjecture of \cite{briet} can be resolved in the affirmative.

\begin{theorem}\label{thm:norndelim}
There exist absolute constants~$c,\varepsilon \in (0,\infty)$ such that for every positive integer~$n$ that is divisible by 8, 
we have 
\[
\ork(\genHam_2^n(n/2))
\geq
2^{\varepsilon n - c}.
\]
\end{theorem} 

\begin{remark}
Our proof of \cref{thm:norndelim} indirectly establishes the result by bounding the Lovász theta number $\lovtheta(\overline G)$ of the complement of~$G =\genHam_2^n(n/2)$ (see \cref{sec:thetabound})
and using the fact that $\ork(G) \geq \lovtheta(\overline G)$~\cite{briet}.
While writing this note, Amir Yehudayoff brought to our attention an unpublished manuscript of Samorodnitsky's~\cite{sam} where he proves a slightly better lower bound on~$\lovtheta(\overline G)$.
He determines this value almost exactly with the use of cleverly-chosen orthogonal polynomials, giving $\lovtheta(\overline G) \approx 2^{0.19 n}$.
Using more elementary methods, we prove that $\lovtheta(\overline G) \geq 2^{0.0435 n - \log_2(5/2)}$.

Based on \cref{thm:symmetric} and Samorodnitsky's results, which cover the graphs~$\genHam_2^n(d)$ for all $d\in[n]$, the best bounds on the orthogonal rank of $G = \genHam_2^n(d)$ can be summarized as follows:
\[
2^{(1-\binent(d/(2n)))n - o(n)} \leq \lovtheta(\overline{G}) \leq \ork(G) \leq \symork(G) = 2^{n-d},
\]
where $h$ denotes the binary entropy function defined above.
\end{remark}

\paragraph{Connection with $k$-wise independence.} Before going into the proofs, we would like to mention a connection between the Lovász theta number of the graph $\genHam_2^n(d)$ and $(d-1)$-wise independent distributions on~$\{0,1\}^n$. 
A probability distribution~$P$ on $\{0,1\}^n$ is \emph{$k$-wise independent} if for any $k$ indices $i_1 < i_2 < \cdots < i_k$ and any string $v\in \{0,1\}^k$
\[
\Pr_{x\sim P}[x_{i_1} x_{i_2} \cdots x_{i_k} = v] = 2^{-k}.
\]
In words, the restriction of $P$ to any $k$ indices is a uniform distribution. Now view $P$ as a function $\{0,1\}^n \to \RR$ such that $P(x) \geq 0$ for all $x\in \{0,1\}^n$ and $\sum_x P(x) = 1$. It is a standard and easy fact that $P$ being $k$-wise independent is equivalent to the Fourier coefficients $\widehat{P}(z)$ being zero for all $z\in \{0,1\}^n$ with Hamming weight $\abs[0]{z} \in \{1,2, \ldots, k\}$. Let $G = \genHam_2^m(d)$. With the above observation and Equation \eqref{eq:prob} on page \pageref{eq:prob} one can prove that the value of
\[
\max\quad P[00\cdots 0] \quad\textnormal{s.t.}\quad\begin{minipage}[t]{7cm}\vspace{-1em}\begin{enumerate}
\item $P$ is a prob.~distr.\ on $\{0,1\}^n$
\item $P$ is $(d-1)$-wise independent
\end{enumerate}
\end{minipage}
\]
is exactly $2^{-n}\lovtheta(\overline{G})$. The maximal probability that all bits are zero was studied in \cite{peled2011maximal} and \cite{benjamini2012k} and it was stated as an open problem in \cite{peled2011maximal} to determine this value for all $d \in [n/2, n)\cap \NN$. Our proof of \cref{thm:norndelim} gives a nontrivial lower bound for $d = n/2$. The results of \cite{sam} yield the asymptotically tight value $2^{-h(d/(2n))n}$ for any $d\in[n]$.

A concise way to phrase the above in Fourier analytic terms is as follows. For a function $f : \bset{n} \to \RR$ and $p\in [1,\infty)$, the $\ell_p$-norm of $f$ is defined as
\[
\norm[0]{f}_{p} \coloneqq  \Bigl(2^{-n} \sum_{\mathclap{x\in\bset{n}}} \abs[0]{f(x)}^p\Bigr)^{1/p}.
\]
Define $\norm[0]{f}_{\infty} \coloneqq \max_{x\in\bset{n}} \abs[0]{f(x)}$.
The above lower  is then equivalent to the assertion that for any $d \in [n]$ and any function $f : \bset{n} \to \RR$ of polynomial degree $d$, we have
\[
\norm[0]{f}_{\infty} \leq 2^{(1-h(d/(2n))) n + o(n)} \norm[0]{f}_{1}.
\]
This may be compared with the following standard consequence of the hypercontractive inequality~\cite[Corollary~5.16]{boucheron2013concentration}, which says that for any function $f:\bset{n} \to \RR$ of polynomial degree $d$ and for all $1<p<q<\infty$, 
\[
\norm[0]{f}_q \leq \Bigl( \frac{q-1}{p-1} \Bigr)^{d/2} \norm[0]{f}_p.
\]

\section{The symmetric orthogonal rank}

In this section we prove \cref{thm:symmetric}. Let us first review some results on the character group of a finite group. Let $\Gamma$ be a finite group and let $\CC^{\times}$ be the multiplicative group $\CC\setminus\{0\}$. The character group $\widehat{\Gamma}$ of $\Gamma$ is the group consisting of all homomorphisms $G \to \CC^{\times}$, that is, maps $f: G\to\CC^{\times}$ such that $f(gh) = f(g)f(h)$ for any $g,h\in \Gamma$. Now consider the complex vector space $\CC^{\Gamma}$ consisting of all maps $\Gamma\to \CC$. Endow this space with the inner product defined by $\langle f,f'\rangle \coloneqq |\Gamma|^{-1}\sum_{g\in G} f(g)f'(g)$. Then the characters $\widehat{\Gamma}$ form an orthonormal basis of $\CC^{\Gamma}$. We can thus write every map $f: \Gamma\to \CC$ in the form 
\[
f(g) = \sum_{\chi \in \widehat{\Gamma}} \widehat{f}(\chi)\, \chi(x)
\]
with $\widehat{f}(\chi) \in \CC$. The complex numbers $\widehat{f}(\chi)$ are called Fourier coefficients and the map $f \mapsto \widehat{f}$ is called the Fourier transform.

Let $m\in \NN$ and let $\zeta_m \in \CC^{\times}$ be an $m$th primitive root of unity. Let $\Gamma$ be the cyclic group $C_m = \{0,1, \ldots, m-1\}$. Then the character group $\widehat{\Gamma}$ consists of the maps
\[
\chi_z : C_m \to \CC^{\times} : x \mapsto (\zeta_m^{z})^{x} \qquad \textnormal{with  $z\in C_m$}.
\]
Let $n\in \NN$ and let $\Gamma$ be the direct power $C_m^{\times n}$. Then the character group $\widehat{\Gamma}$ consists of the maps
\[
\chi_z : C_m^{\times n} \to \CC^{\times} : x \mapsto (\zeta_m^{z_1})^{x_1} \cdots (\zeta_m^{z_n})^{x_n} \qquad \textnormal{with  $z\in C_m^{\times n}$}.
\]
We will write the product $(\zeta_m^{z_1})^{x_1} \cdots (\zeta_m^{z_n})^{x_n}$ as $\zeta_m^{z\cdot x}$, and $\widehat{f}(\chi_z)$ as $\widehat{f}(z)$.

We will use Bochner's theorem for finite groups. Let $f$ be a map $\Gamma \to \CC$. Let $e$ be the unit element of $\Gamma$. We say $f$ is \emph{normalized} if $f(e) = 1$. We say that~$f$ is \emph{positive semidefinite} (PSD) if for any $k\in \NN$ and any $g_1, \ldots, g_k\in \Gamma$ the matrix $(f(g^{\vphantom{-1}}_i\hspace{0.1em} g_j^{-1}))_{i,j\in [k]}$ is PSD.

\begin{theorem}[Bochner's theorem for finite groups]\label{bochner}
Let $\gengrp$ be a finite abelian group. Let $\widehat{\gengrp}$ be the character group of $\gengrp$. Let $f$ be a map $\gengrp\to \CC$. Then, the following two statements are equivalent:
\begin{enumerate}
\item The map $f$ is normalized and PSD. 
\item The map $f$ satisfies $\widehat{f}(\chi)\in \RR_{\geq0}$ for all $\chi \in \widehat{\gengrp}$, and $\sum_{\chi \in \widehat{\gengrp}} \widehat{f}(\chi) = 1$.
\end{enumerate}
\end{theorem}
\begin{proof}
Assume $f$ is normalized and PSD. Consider the Fourier decomposition $f = \sum_{\chi \in \widehat{\gengrp}} \widehat{f}(\chi)\chi$. Then $f(xy^{-1}) = \sum_{\chi}\widehat{f}(\chi) \chi(x) \overline{\chi(y)}$. Define the matrix $M = (f(gh^{-1}))_{g,h\in \gengrp}$. Then for any vector $v\in\CC^\Gamma$, we have $v^*Mv \in \RR_{\geq 0}$ and
\begin{align*}
v^* M v &= v^*\Bigl( \sum_\chi \widehat{f}(\chi)(\chi(g)\overline{\chi(h)})_{g,h} \Bigr) v\\
&= \sum_{\chi} \widehat{f}(\chi) v^* (\chi \chi^*) v,
\end{align*}
where in the last line we used $\chi$ to denote the complex vector $(\chi(g))_g$.
By taking $v = \chi$ and using the orthogonality of the characters, we get $\widehat{f}(\chi)\in \RR_{\geq 0}$ for all $\chi$. Also $1 = f(e) = \sum_{\chi} \widehat{f}(\chi) \chi(e) = \sum_{\chi} \widehat{f}(\chi)$.

Assume $f$ has real nonnegative Fourier coefficients summing to 1. Then $f(e) = \sum_{\chi} \widehat{f}(\chi) \chi(e) = \sum_{\chi} \widehat{f}(\chi) = 1$. Let $g_1,\ldots, g_k \in \Gamma$. Define the matrix $M = (f(g^{\vphantom{-1}}_i\hspace{0.1em} g_j^{-1}))_{i,j\in[k]}$. Then ${M = \sum_{\chi} \widehat{f}(\chi)(\chi(g_i) \overline{\chi(g_j)})_{i,j\in[k]}= \sum_{\chi} \widehat{f}(\chi) N_\chi}$,
where each $N_\chi$ is a submatrix of the PSD matrix $\chi \chi^*$. Therefore, the matrix~$M$ is PSD.
\end{proof}

The following proposition relates symmetric orthogonal embeddings to maps $f:\gengrp\to \CC$ with restrictions on the Fourier coefficients.

\begin{proposition}\label{bochnerappl}
Let $\gengrp$ be a finite abelian group and let $S$ be a subset of $\gengrp$. Let $f: \gengrp \to \CC$ be a map such that $f(e) = 1$ and $f(g) = 0$ for all $g\in S$. Then, there exists a map $\phi : \gengrp \to \CC^d$ such that $\langle \phi(g), \phi(h)\rangle =  f(gh^{-1})$ if and only if all Fourier coefficients $\widehat{f}(\chi)$ are real and nonnegative and $\abs[0]{\supp(\widehat{f})} \leq d$.
\end{proposition}
\begin{proof}
Let $\phi: \gengrp \to \CC^d$ be a map such that $\langle \phi(g), \phi(h)\rangle = f(gh^{-1})$.
We have the following equality of matrices
\[
(\langle \phi(g), \phi(h)\rangle)_{g,h\in \Gamma} = (f(gh^{-1}))_{g,h\in \Gamma} = \sum_{\chi} \widehat{f}(\chi)\chi \chi^*.
\]
The left-hand side is a Gram matrix and therefore PSD. Bochner's theorem (\cref{bochner}) says that the Fourier coefficients $\widehat{f}(\chi)$ are then real and nonnegative. Moreover, the rank of the left-hand side is at most~$d$ while the rank of the right-hand side equals $\abs[0]{\supp (\widehat{f})}$.

On the other hand, suppose $\widehat{f}(\chi)\in \RR_{\geq 0}$ for all $\chi$. 
Let $S$ be the set $\{\chi \in \widehat{\gengrp} : \widehat{f}(\chi) \neq 0\}$. 
For any $g\in \gengrp$, define the vector
\[
\phi(g) \coloneqq \Bigl(\sqrt{\widehat{f}(\chi)}\, \chi(g)\Bigr)_{\!\!\chi\in S} \in \CC^{S}.
\]
We claim that $\phi$ satisfies $\langle \phi(g), \phi(h)\rangle = f(gh^{-1})$ for all $g,h\in \Gamma$. Indeed, we have, for any $g,h \in \gengrp$,
\begin{align*}
\langle\phi(g), \phi(h)\rangle &= \sum_{\chi \in S} \widehat{f}(\chi) \chi(g)\overline{\chi(h)} = f(gh^{-1}) = 0, \quad \textnormal{when $gh^{-1} \in S$},\\
\langle\phi(g), \phi(g)\rangle &= \sum_{\chi \in S} \widehat{f}(\chi) \chi(g)\overline{\chi(g)} = f(e) = 1,
\end{align*}
which proves the claim.
\end{proof}

We also use the following well-known result on the number of roots of multivariate polynomials, the particular form of which is taken from~\cite{Cohen:2015}. View the cyclic group $C_m$ as a multiplicative subgroup of $\CC$ by mapping a generator to a primitive $m$th root of unity. For any map $f : C_m^{\times n} \to \CC$, we define the \emph{polynomial degree} $\deg(f)$ to be the smallest number $d$ such that there is a polynomial $p\in \CC[x_1, \ldots, x_n]$ of degree $d$ that interpolates $f$, that is, $p(z) = f(z)$ for all $z\in C_m^{\times n}$.
We write $U(f) := \{z\in C_m^{\times n}\mid f(z)\ne 0\}$ for the set of nonzeros of~$f$ in~$C_m^{\times n}$.

\begin{theorem}[{DeMillo-Lipton-Schwartz-Zippel}]\label{SchwartzZippel}
Let $f:C_m^{\times n} \to \CC$ be a nonzero map of polynomial degree $d$. 
Then, 
\[
|U(f)|\geq
\frac{m^{\smash{n}}}{m^{d/(m-1)}}.
\]
\end{theorem}
\begin{proof}
By viewing $C_m$ as a multiplicative subgroup of $\CC$ we can identify~$f$ with a nonzero polynomial in $\CC[x_1,\ldots,x_n]$ of degree $d$ such that each variable in $f$ has degree at most $m-1$. 
We induce on~$n$.
For the base case $n=1$, $f$ is a nonzero univariate polynomial of degree at most $m-1$ and $f$ thus has at most $m-1$ zeros. Therefore, $\abs[0]{U(f)} \geq 1 \geq m/m^{d/(m-1)}$.

Assume the theorem statement is proven for polynomials in $n-1$ variables. We can write $f$ in the form
\[
f(t, y_1, \ldots, y_{n-1}) = \sum_{i=1}^{\smash{d}} t^i g_i(y_1,\ldots, y_{n-1}),
\]
with $g_i \in \CC[y_1,\ldots,y_{n-1}]$ a polynomial of degree at most $d-i$. Let $k$ be the maximum $i$ for which $g_i$ is nonzero. By the induction hypothesis, the polynomial $g_k$ satisfies
\[
\abs[0]{U(g_k)} \geq m^{n-1}/m^{(d-k)/(m-1)}.
\]
For each $y\in U(g_k)$, let $h_y \in \CC[t]$ be the univariate polynomial defined by $h_y(t) = f(t,y_1,\ldots, y_{n-1})$.
We know that each $h_y$ is nonzero and has degree~$k$. Therefore, $\abs[0]{U(h_y)} \geq m/m^{k/(m-1)}$. We conclude that
\[
\abs[0]{U(f)} \geq \sum_{\smash{\mathclap{y \in U(g_k)}}} \abs[0]{U(h_y)} \geq m^n/m^{d/(m-1)},
\]
which proves the theorem.
\end{proof}

\begin{proof}[\bf\upshape Lower bound proof for \cref{thm:symmetric}] For $x\in C_m^{\times n}$, define the weight $\abs[0]{x}$ to be $\sum_{i}x_i$ where the sum is taken in $\NN$.
Denote the unit element in $C_m^{\times n}$ by 0.
Let $f: C_m^{\times n} \to \CC$ be a map satisfying the three properties
\begin{enumerate}
\item $f(0) = 1$,
\item $f(x) = 0$ when $|x| \geq d$,
\item $\widehat{f}(z) \geq 0$ for all $z\in C_m^{\times n}$.
\end{enumerate}
Write $f$ in the Fourier basis, $f = \sum_{z\in C_m^{\times n}} \widehat{f}(z) \chi_z(x)$. Define $g: C_m^{\times n} \to \CC$ by $g(z) = \widehat{f}(z)$. Then $\widehat{g}(z) = m^{-n} f(z)$. 
The Fourier expansion of $g$ is thus $g(x) = m^{-n} \sum_{z\in C_m^{\times n}} f(z) \chi_z(x)$. Since $f(x) = 0$ when $\abs{x} \geq d$, the polynomial degree of $g$ is at most $d$. By \cref{SchwartzZippel} there are at least $m^n/m^{d/(m-1)}$ points where $g$ is nonzero. The map $f$ thus has at least $m^n/m^{d/(m-1)}$ nonzero Fourier coefficients, which means by \cref{bochnerappl} that $\symork(G) \geq m^{n - d/(m-1)}$.
\end{proof}

\begin{proof}[\bf\upshape Upper bound proof for \cref{thm:symmetric}]
Define $k = n - d/(m-1)$.
Let $h: C_m^{\times k} \to \CC$ be the indicator function of the zero element in $C_m^{\times k}$, so $h(0) = 1$ and $h(x) = 0$ for all $x\neq 0$.
Let $g : C_m^{\times n} \to \CC$ map $x$ to $h(x_1,\ldots, x_{k})$. 
We see that $g(0) = 1$ and that $g(x) = 0$ whenever $|x| \geq d$. 
For the Fourier coefficients of $g$ we get
\[
\langle g, \chi_z \rangle = \frac{1}{m^n} \sum_{x\in C_m^{\times n}} g(x) \chi_z(x) = \frac{1}{m^n} \sum_{\substack{x\in C_m^{\times n}:\\ x = 0^{k} x'}} \zeta_m^{z'\cdot x'},
\]
where $z' \in C_m^{\times d/(m-1)}$ is the vector consisting of the last $d/(m-1)$ entries in~$z$, and similarly for $x'$. If $z' = 0$, then the above expression is positive. If $z'\neq 0$, then the above expression is zero. We conclude that $g$ has $m^{k}$ nonzero Fourier coefficients and moreover all nonzero Fourier coefficients are positive. Therefore, by \cref{bochnerappl}, there is a symmetric orthogonal embedding of the graph $\genHam_m^n(d)$ in dimension $m^{k}$.
\end{proof}


%
\section{Lower bounds on the orthogonal rank}
\label{sec:thetabound}

In this section we prove \cref{thm:norndelim}. The workhorse in this proof is the following special case of \cite[ Lemma~3.3]{sam}, for which we give an alternative, more elementary proof, albeit with slightly worse constants.



\begin{proposition}\label{simpler}
There exist absolute constants~$c,\varepsilon \in (0,\infty)$ such that for every~$n\in\NN$ that is divisible by~$8$ and any 
 $p\in \RR[x]$ of degree $<n/2$, we have
\[
p(0) \leq 2^{-\varepsilon n+c} \sum_{i=0}^n \binom{n}{i}\abs[0]{p(i)}.
\]
\end{proposition}
\begin{proof}
Let $S\subseteq \{0,1,\ldots, n\}$ be any set of size at least $n/2$. Writing $p$ as an interpolation polynomial in Lagrange form with respect to the set $S$, gives
\[
p(x) = \sum_{i\in S} p(i) \prod_{\ell\in S\setminus\{i\}} \frac{\ell-x}{\ell-i}.
\]
By the triangle inequality, $p(0)$ is therefore at most
\begin{align*}
p(0) &\leq \sum_{i\in S} \abs[0]{p(i)} \!\prod_{\ell\in S\setminus\{i\}}  \frac{\ell}{\abs[0]{\ell-i} }\\
&= \sum_{i\in S} \abs[0]{p(i)} \binom{n}{i}\binom{n}{i}^{\!-1} \!\!\!\prod_{\ell\in S\setminus\{i\}} \frac{\ell}{\abs[0]{\ell-i} }\\
&\leq \biggl(\sum_{i=0}^n\binom{n}{i}\abs{p(i)}\biggr)\: \max_{i\in S} \binom{n}{i}^{\!-1} \!\!\!\prod_{\ell\in S\setminus\{i\}} \frac{\ell}{\abs[0]{\ell-i} }.
\end{align*}
Hence,
\[
p(0)
\leq
\biggl(\sum_{i=0}^n\binom{n}{i}\abs{p(i)}\biggr)\: \min_S\max_{i\in S} \binom{n}{i}^{\!-1} \!\!\!\prod_{\ell\in S\setminus\{i\}} \frac{\ell}{\abs[0]{\ell-i} },
\]
where the minimum is taken over all sets $S\subseteq \{0,1,\dots,n\}$ of size at least~$n/2$.
To prove the result it thus suffices to exhibit a set~$S$ for which the maximum is exponentially small in~$n$.
To this end, define the sets $S_1 \coloneqq (\tfrac{n}{8},\tfrac{3n}{8}]\cap \NN$ and $S_2 \coloneqq [\tfrac{5n}{8},\tfrac{7n}{8}) \cap \NN$ and let $S = S_1 \cup S_2$. 
Define 
\begin{equation}\label{fdef}
f(i) \coloneqq \binom{n}{i}^{\!-1}\!\!\! \prod_{\ell\in S\setminus\{i\}} \frac{\ell}{\abs[0]{\ell-i} }. 
\end{equation}
Since~$f$ is symmetric about $n/2$, we have $\max_{i\in S_1} f(i)=\max_{i\in S_2} f(i)$ and it follows that to bound $\max_{i\in S}f(i)$, it is sufficient to maximize $f$ over $S_1$. 

Define $k = n/8$. Let $j\in [1,3]$ such that $i \coloneqq jk \in S_1$. 
We claim that
\begin{equation}\label{eq:fbinom}
f(i) \leq \frac{\binom{3k}{jk}\binom{7k-1}{jk}\binom{jk}{k}}{\binom{5k-1}{jk}\binom{8k}{jk}}.
\end{equation}
Indeed, in \eqref{fdef} we can split the product over $S$ into a product over $S_1$ and a product over $S_2$ to obtain
\begin{align*}
f(j k) &= \binom{8k}{jk}^{\!-1} \!\!\!\!\prod_{\ell\in S_1\setminus\{jk\}} \frac{\ell}{|\ell - jk|} \prod_{\ell\in S_2} \frac{\ell}{|\ell - jk|}\\
&\leq \binom{8k}{jk}^{\!-1} \frac{ (3k)! } { k!\, ((3-j)k)!\, ((j-1)k)! } \frac{ (7k-1)!\, ((5-j)k-1)!}{(5k-1)!\, ((7-j)k-1)!}.
\end{align*}
Multiplying with $(jk)!^2/(jk)!^2$ and grouping appropriately, one recognizes the required binomial coefficients.


Observe that the product of the first and third coefficients in the numerator of~\eqref{eq:fbinom}, namely $\binom{3k}{jk}\binom{jk}{k}$, counts the number of ways to choose a $jk$-subset in a $3k$-set and then a $k$-subset in this $jk$-subset. We get the same count by first choosing a $k$-subset in a $3k$-set and then choosing a $jk$-subset in the $3k$-set which includes the $k$-set. Therefore, $\binom{3k}{jk}\binom{jk}{k} = \binom{3k}{k} \binom{3k-k}{jk-k} = \binom{3k}{k}\binom{2k}{jk-k}$. 
Next, it follows from the Cauchy–Vandermonde identity~\cite[Exercise~1.9]{jukna2011extremal},
\[
{m + n\choose r} = \sum_{k=0}^r{m\choose k}{n\choose r- k},
\]
that $\binom{3k}{k} \binom{2k}{jk-k} \leq \binom{5k}{jk} = \frac{5}{5-j}\binom{5k -1}{jk}$.
Hence,
$
f(i) \leq \frac{5}{5-j}\binom{7k-1}{jk}\binom{8k}{jk}^{-1}$.
Finally, since $j\leq 3$ and  for any integers
$b\leq b+c<a$, we have $\binom{a-c}{b}\binom{a}{b}^{-1} \leq (\frac{a-b}{a})^c$ \cite[Exercise 1.18]{jukna2011extremal}, it follows that
\begin{align*}
f(i) &\leq
 \frac{5}{5-j}\binom{7k-1}{jk}\binom{8k}{jk}^{-1}\\
 &\leq
 \frac{5}{2}\binom{7k}{jk}\binom{8k}{jk}^{\!-1}\\ 
 &\leq 
 \frac{5}{2}\Bigl(\frac{8k-jk}{8k}\Bigr)^k 
\leq
\frac{5}{2}\Bigl(\frac{7}{8}\Bigr)^k, 
\end{align*}
which establishes the result.
\end{proof}


For a matrix $X \in \RR^{n\times n}$, we write $X\succcurlyeq 0$ if $X$ is symmetric and PSD.

\begin{proof}[\bf\upshape Proof of \cref{thm:norndelim}]
Let $G = \genHam_2^n(d)$.
We lower bound $\ork(G)$ by lower bounding $\lovtheta(\overline{G})$. This we do by looking at the value of $\lovtheta(G)$. By definition,
\begin{alignat}{2}
\lovtheta(G) &= \max\quad \sum_{i,j\in[2^n]} X_{ij} \quad&\textnormal{s.t.}&\quad\begin{minipage}[t]{6cm}\vspace{-1em}\begin{enumerate}
\item $X$ is a real $2^n\times 2^n$ matrix
\item $X\succcurlyeq0$
\item $\sum_{i=1}^{2^n} X_{ii} = 1$
\item $X_{ij} = 0\quad \forall (i,j)\in E(G)$.
\end{enumerate}
\end{minipage}\nonumber
\intertext{If $X$ is a feasible solution to the above maximisation, then for any element $a\in G$ (we implicitly identify $G$ with the group $C_2^{\times n}$ here) the matrix $Y^a$ defined by $Y^a_{x,y} = X_{x+a, y+a}$ is a feasible solution with the same value. Let~$Y$ be the average $\frac{1}{2^n}\sum_{a\in G} Y^a$. This is again feasible with the same value. Moreover $Y_{x,y}$ depends only on $x-y$; namely, if $x-y = x'-y'$, then $Y_{x,y} = Y_{x-y,0} = Y_{x'-y',0} = Y_{x',y'}$. With this observation and Bochner's Theorem (\cref{bochner}) we obtain}
\lovtheta(G) &= \max\quad \sum_{\mathclap{x\in\{0,1\}^n}} f(x) \quad&\textnormal{s.t.}&\quad\begin{minipage}[t]{6cm}\vspace{-1em}\begin{enumerate}
\item $f:\{0,1\}^n \to \RR$
\item $\widehat{f}(z) \geq 0$ for all $z\in \{0,1\}^n$
\item $f(0) = 1$
\item $f(x) = 0$ for $d \leq |x|$.
\end{enumerate}
\end{minipage}\nonumber
\intertext{Through the Fourier transform we get}
\lovtheta(G) &= \max \quad 2^n g(0) \quad&\textnormal{s.t.}&\quad\begin{minipage}[t]{6cm}\vspace{-1em}\begin{enumerate}
\item $g:\{0,1\}^n \to \RR$
\item $g(z) \geq 0$ for all $z$
\item $2^n\,\widehat{g}(0) = 1$
\item $\widehat{g}(x) = 0$ for $d \leq |x|$.
\end{enumerate}
\end{minipage}\label{eq:prob}
\intertext{(See \cite{schrijver1979comparison} for a similar description.) Because $\lovtheta(\overline{G})\lovtheta(G) \geq 2^n$ (see \cite{lovasz1979shannon}), we have}
\lovtheta(\overline{G}) &\geq \frac{2^n}{\max\, 2^n g(0)} \quad&\textnormal{s.t.}&\quad\begin{minipage}[t]{6cm}\vspace{-1em}\begin{enumerate}
\item $g:\{0,1\}^n \to \RR$
\item $g(z) \geq 0$ for all $z$
\item $2^n\,\widehat{g}(0) = 1$
\item $\widehat{g}(x) = 0$ for $d \leq |x|$.
\end{enumerate}
\end{minipage}\nonumber
\end{alignat}
Now let $d = n/2$. According to \cref{simpler}, the value of $g(0)$ is at most~$2^{-\varepsilon n+c}$, so the value of $\lovtheta(\overline{G})$ is at least $2^{\varepsilon n-c}$.
\end{proof}

\paragraph{Acknowledgements.} The authors thank Harry Buhrman, Teresa Piovesan, Oded Regev, Ronald de Wolf, and Amir Yehudayoff for helpful discussions.  

\newpage

\raggedright
\bibliographystyle{alphaabbr}
\bibliography{all}

\end{document}